\documentclass[11pt]{article}

\usepackage{amsfonts,amssymb,amsmath,latexsym,bbm,ae,aecompl,caption}

\newcommand{\FFF}{\vspace*{6pt}}
\newcommand{\B}{\vspace*{-2pt}}
\newcommand{\BB}{\vspace*{-4pt}}

\newcommand{\Paragraph}[1]{\BB\BB\BB\B\paragraph{#1}}
\newcommand{\remove}[1]{}

\newcommand{\cA}{{\mathcal A}}
\newcommand{\cO}{{\mathcal O}}

\newcommand{\cE}{{\mathcal E}}

\newcommand{\qed}{\hfill $\square$ \smallskip}

\newenvironment{proof}{\noindent{\bf Proof:}}{\qed}

\newtheorem{theorem}{Theorem}
\newtheorem{lemma}{Lemma}

\newlength{\pagewidth}
\begin{document}

\baselineskip   		3ex
\parskip        		1ex

\title{			Adversarial Multiple Access Channels \\ with Individual Injection Rates 				\footnotemark[1]  \vfill}

\author{		Lakshmi Anantharamu 	\footnotemark[2] \and
			Bogdan S. Chlebus  		\footnotemark[3] \and
			Mariusz A. Rokicki 		\footnotemark[4] }

\footnotetext[1]{The results of this paper were presented in a preliminary form in~\cite{AnantharamuCR-OPODIS09} and in its final journal form in~\cite{AnantharamuCR-TCS17}. 
The work of the first author was supported by NSF Grant 1016847.
The work of the second author was supported in part by NSF Grants 0310503 and 1016847, and by the Engineering and Physical Sciences Research Council under [grant number EP/G023018/1].
The work of the third author was partly done when he was with the University of Colorado Denver, and it was supported by NSF Grant 0310503 and by the Engineering and Physical Sciences Research Council  under [grant number EP/G023018/1].}

\footnotetext[2]{Department of Computer Science and Engineering, 
			University of Colorado Denver, 
			Denver, Colorado 80217, U.S.A.}

\footnotetext[3]{Department of Computer Science and Engineering, 
			University of Colorado Denver, 
			Denver, Colorado 80217, U.S.A.}

\footnotetext[4]{Department of Computer Science, 
			University of Liverpool, 
			Liverpool L69 3BX, United Kingdom.}

\date{}

\maketitle

\vfill

.
\begin{abstract}
We study deterministic distributed broadcasting in synchronous multiple-access channels. 
Packets are injected into $n$ nodes by a window-type adversary that is constrained by a window~$w$ and injection rates  individually assigned to all nodes. 
We investigate what queue size and packet latency can be achieved with the maximum aggregate injection rate of one packet per round, depending on properties of channels and algorithms.
We give a non-adaptive  algorithm for channels with collision detection and an adaptive  algorithm for channels without collision detection that achieve $\cO(\min(n+w,w\log n))$ packet latency.
We show that packet latency has to be either $\Omega(w \max (1,\log_w n))$, when $w\le n$, or $\Omega(w+n)$, when $w>n$, as a matching lower bound to these algorithms.
We develop a non-adaptive algorithm for channels without collision detection that achieves $\cO(n+w)$ queue size and $\cO(nw)$ packet latency.
This is in contrast with the adversarial model of global injection rates, in which non-adaptive algorithms with  bounded packet latency do not exist (Chlebus et al. \emph{Distributed Computing} 22(2): 93--116, 2009).
Our algorithm avoids collisions produced by simultaneous transmissions; we show that any algorithm with this property must have $\Omega(nw)$ packet latency.

\vfill

\end{abstract}

\vfill

\thispagestyle{empty}

\setcounter{page}{0}

\newpage

\section{Introduction}

\label{sec:introduction}

A multiple access channel is a model of a communication environment supporting broadcasting among a set of nodes. 
What defines such a network is the property that a transmission by a node is heard by every node in the system precisely when this transmission does not overlap with other transmissions. 
Overlapping multiple transmissions collide with one another so that none can be received successfully.

There are two popular representations of multiple-access channels.
One representation is provided by single-hop radio networks.
These are wireless networks in which nodes use one radio frequency and a transmission by a node can be sensed by any other node. 
In such networks, overlapping transmissions interfere with one another, which results in receiving the contents attempted to be communicated as garbled, while a single transmission is successfully received by every node.
Another representation is provided by the implementation of local area networks by a wired Ethernet, as represented by the IEEE 802.3 collection of standards.
Abstracting multiple access channels from the medium access control of such communication environments allows one to study the optimality of communication algorithms in an idealized but precisely defined algorithmic communication setting, without the constraints of continually evolving wireless technologies and IEEE standards. 

We  consider the synchronous slotted model of communication, in which executions are partitioned into rounds determined by a global clock.
When a transmission is successful, then the transmitted message is immediately received by every other node attached to the multiple-access channel.
Messages transmitted in one round by different nodes are considered as overlapping and so colliding with one another.
With these stipulations, at most one message can be successfully broadcast in a round.

Performance of broadcast algorithms can be measured in terms of various natural metrics.
Among them, queue size and packet delay are the most often considered. 
Stability of a system means that the number of packets stored in local queues is bounded in all rounds.
An apparently weakest expectation about the amount of time spent by packets in the system is that every packet is eventually successfully transmitted, see~\cite{ChlebusKR09} and~\cite{RosenT06}.
Packet latency denotes the maximum time that a packet may spend in a queue.
Even when every packet is eventually successfully transmitted, there may be no finite bound on the times that packets spend waiting in queues; when a finite bound exists then the system is stable.

It is natural to consider packet generation as restricted by two parameters: (1) an average frequency of injection at nodes and (2) an upper bound on the number of packets that can be injected in one round.
The frequency of packet injections into the system is called injection rate.
Such frequencies are typically determined by the statistical constraints when modeling and simulating networks; in this paper we do not adhere to such approaches.
The number of packets that can be injected into the system in one round is often called the burstiness of traffic.

Using randomization is a natural means to implement efficient arbitration for access to the channel; popular randomized algorithms include Aloha \cite{Abramson85} and backoff ones \cite{BenderFHKL05,GoldbergJKP04,HastadLR96,MetcalfeB76}.
We consider deterministic algorithms that do not resort to randomization.
We use adversarial queuing as the methodology to study the worst-case performance of these deterministic communication algorithms.
Adversarial queuing is based on defining packet injection rates as upper bounds on the average number of packets inserted into the system.
The averages of the numbers of injections are defined without resorting to stochastic notions.

Recent work on adversarial queuing for multiple access channels, see~\cite{AnantharamuCKR-INFOCOM10,ChlebusKR09,ChlebusKR-TALG}, concerned adversaries constrained by ``global'' injection rates, in the sense that the adversary is restricted by the number of packets injected into all nodes but not by how many packets are injected into any specific node.
We consider adversaries with individual rates associated with nodes; in this model the adversary's injection rate is constrained separately for each node.
An adversary with injection rates associated with individual nodes is more restricted than one with a ``global'' injection rate.
We understand the throughput of an algorithm as the maximum (global) rate for which it is stable.
One may observe that the (global) injection rate of one packet per round is the maximum rate that allows for a stable algorithm.
It follows that the throughput of any algorithm is at most~$1$.

A definition of injection rates depends on how we average injections.
There are two popular approaches to define averaging, which produce two corresponding classes of  adversaries, called leaky-bucket and window ones.
The former class comprises general adversaries for which all the time intervals are used to constrain injection rate by averaging over them.
The latter class restricts intervals over which averages are taken to be of the same ``window'' length, which is a parameter of an adversary.
When injection rates are smaller than~$1$ then the two adversarial models of ``global'' injection rates are equivalent, as shown by Ros{\'e}n~\cite{Rosen02}.

In the context of adversarial queuing in store-and-forward networks, ``globally constrained'' window adversaries were first used by Borodin et al. \cite{BorodinKRSW01} and similar leaky-bucket ones by Andrews et al. \cite{AndrewsAFLLK01}.
For the multiple access channel, Chlebus et al.~\cite{ChlebusKR-TALG} used window adversaries with injection rates smaller than~$1$, which does not restrict the generality of results for such rates.
Chlebus et al.~\cite{ChlebusKR09} used both models for ``global'' injection rate~$1$; for such highest possible ``global'' rate, the window adversarial model is strictly weaker, see \cite{ChlebusKR09,Rosen02}.

We consider constraints on adversaries formulated as restrictions on what  can be injected separately at each node and also by what can be injected into all the nodes combined. 
Constraints on adversaries are \emph{global} if they are formulated in terms of  the numbers of packets injected in suitable time intervals, without any concern about the nodes in which the packets are injected.
Traffic constraints are \emph{local} when the patterns of injection are considered separately  for each node.
Constraints for the whole network implied by local ones are called \emph{aggregate} in this paper.
In particular, we may have local  and aggregate and global burstiness, and local and aggregate and global injection rates.
Observe that global constraints are logically weaker than aggregate constraints.
Global and aggregate injection rate~$1$ is the maximum that a channel can handle in a stable way,  since stability provides that the throughput rate is at least as large as the injection rate.
Adversaries with global injection rates were used in~\cite{ChlebusKR09,ChlebusKR-TALG} to model traffic requirements in multiple access channels.
In this paper, we introduce adversaries with local injection rates.

To categorize algorithms, we use the terminology similar to that in in~\cite{ChlebusKR09,ChlebusKR-TALG}.
Algorithms may either use control bits piggybacked on transmitted packets or not; when they do,  then they are called adaptive.


\Paragraph{A summary of the results.}

We study effectiveness of broadcasting when traffic demands are specified as adversarial environments with individual injection rates associated with nodes and when algorithms are both distributed and deterministic.
We allow the adversaries to be such that the associated aggregate injection rate (the sum of all the individual rates) is~$1$, which is the maximum value allowing for stability.
This is the first study of adversaries injecting packets into multiple-access channels and restricted this way by individual rates, to the authors' knowledge.

A general goal is to explore what quality of broadcast can be achieved for individual injection rates. 
In particular, we want to compare adversarial environments defined by individual rates with those determined by global ones,  under the maximum average broadcast load of one packet per round.
The underlying motivation for this work was that individual injection rates are more realistic in moderate time spans and hopefully the limitations on the quality of broadcast with throughput~$1$ discovered in~\cite{ChlebusKR09} could be less severe when the rates are individual.
Indeed, we show that bounded packet latency is achievable with individual injection rates when the aggregate rate is~$1$.
This is in contrast with the model of global injection rates, for which achieving bounded waiting times is impossible when the throughput equals~$1$, as demonstrated in~\cite{ChlebusKR09}.

We denote the number of stations by $n$ and the adversary's window size by~$w$.

The comparison of the two models of global and individual injection rates for a window-type adversary with respect to possibility/impossibility of obtaining stability or bounded packet latency is summarized in Table~\ref{table:comparison}.
Additional explanation of this table is as follows.
The possibility in a given row means the existence of a broadcast algorithm of a given class (either acknowledgment-based or non-adaptive or adaptive) that always achieves the claimed performance characteristic (either stability or bounded packet latency) in a system of a given size (either $n=2$ or $n=3$ or $n\ge 4$).
The column for the model of global injection rates reflects the results given in~\cite{ChlebusKR09}.


\newcommand{\RB}{\raisebox{3ex}{~}}
\newcommand{\LB}{\raisebox{-1.5ex}{~}}

\begin{table}
\caption{\label{table:comparison}
A comparison of possibilities and impossibilities to achieve performance characteristics, depending on whether a window adversary of the aggregate injection rate~$1$ is constrained by a global injection rate or individual injection rates.}
\begin{center}
\begin{tabular}{l l l l l }
\hline
\RB \LB
class of algorithms&size&quality& global rate & individual rates\\
\hline\hline
\RB \LB
ack-based & $n=2$ & stable
& I & I \\
\hline
\RB \LB
non-adaptive & $n=2$ & bounded latency
& P & P \\
\hline
\RB \LB
non-adaptive & $n=3$ & stable
& I & P \\
\hline
\RB \LB
adaptive & $n=3$ & bounded latency
& P & P \\
\hline
\RB \LB
adaptive & $n\ge 4$ & stable
& P & P \\
\hline
\RB \LB
adaptive & $n\ge 4$ & bounded latency
& I & P \\
\hline
\end{tabular}
\end{center}
\FFF

\caption*{
The parameter $n$ denotes the number of stations.
Letter P denotes possibility and I impossibility.
The impossibility results hold for channels with collision detection and the possibility results hold for channels without collision detection.
}
\end{table}

We further investigate the model of individual injections by deriving bounds on queue size and packet latency.
Acknowledgment-based algorithms cannot achieve throughput~$1$, which strengthens the result for global injection rates~\cite{ChlebusKR09}.
We give a non-adaptive algorithm of $\cO(\min(n+w,w\log n))$ packet latency when collision detection is available.
An adaptive algorithm can achieve similar performance without collision detection; this is because control bits allow to simulate collision detection with a constant overhead per round.
Bounded packet latency can also be achieved by non-adaptive algorithms in channels without collision detection.
More precisely, we develop a non-adaptive algorithm with $\cO(n+w)$ size queues and $\cO(nw)$ packet latency.
The optimality of our non-adaptive  algorithm for channels without collision detection, in terms of packet latency, is left open, but we demonstrate optimality in the class of  algorithms that avoid collisions altogether.


\Paragraph{Previous work.}

Most of the previous work on dynamic broadcasting in multiple-access channels has been carried out under the assumption that packets were injected subject to stochastic constraints.
Such systems can be modeled as Markov chains with stability usually understood as ergodicity, but alternatively, stability may mean that throughput under the given injection rate is as large as the  injection rate.
Gallager~\cite{Gallager85} gives an overview of early work in this direction and Chlebus~\cite{Chlebus-randomized-radio-chapter-2001} surveys later work.
In particular, the popular and early developed broadcast algorithms, like Aloha~\cite{Abramson85} and binary exponential backoff~\cite{MetcalfeB76}, have been extensively studied for stochastic injection rates.
For recent work, see the papers by Bender et al.~\cite{BenderFGY-SODA16}, Goldberg et al.~\cite{GoldbergJKP04,GoldbergMPS00}, H\aa stad et al.~\cite{HastadLR96}, and Raghavan and Upfal~\cite{RaghavanU98}.
Stability of backoff algorithms for multiple-access channels with packet injection controlled by a window adversary was studied by Bender et al.~\cite{BenderFHKL05} in the queue-free model; they showed that the exponential backoff is stable for $\cO(1/\log n)$ arrival rates and it is unstable for arrival rates that are $\Omega(\log \log n/\log n)$.
Awerbuch et al.~\cite{AwerbuchRS08} developed a randomized algorithm for multiple-access channels competing against an adaptive jamming adversary that achieves a constant throughput for the non-jammed rounds.

Deterministic distributed broadcast algorithms for multiple-access channels in the model with queues were first studied by Chlebus et al.~\cite{ChlebusKR-TALG} according to the methodology of adversarial queuing.
They considered classes of deterministic distributed algorithms, including adaptive and acknowledgment based.
They defined latency to be fair when it was $\cO(\text{burstiness}/\text{rate})$ and considered stability to be strong when queues were $\cO(\text{burstiness})$.
That paper~\cite{ChlebusKR-TALG} gave a non-adaptive algorithm achieving fair latency for $\cO(1/{\text{polylog}\ n})$ rates and showed that no algorithm could be strongly stable for rates that are $\omega(\frac{1}{\log n})$.
They also showed that no oblivious  acknowledgment-based algorithm could be stable for rates larger than $\frac{3}{1+\lg n}$, and hence that there are no universally stable oblivious acknowledgment-based algorithms.
Two oblivious acknowledgment-based algorithms were developed in~\cite{ChlebusKR-TALG}: one of fair latency for rates at most  $\frac{1}{c n\lg^2 n}$, for a sufficiently large $c>0$, and  an explicit one of fair latency for rates at most $\frac{1}{27 n^2\ln n}$.
In the subsequent work on the deterministic distributed algorithms for the multiple access channel with global injection rates, Chlebus et al.~\cite{ChlebusKR09} investigated the quality of broadcast for the maximum throughput, that is, the maximum rate for which stability is achievable.
They defined fairness to mean that each packet is eventually heard on the channel.
They developed a stable algorithm with $\cO(n^2+\text{burstiness})$ queues against leaky-bucket adversaries of injection rate~$1$, which demonstrated that throughput~$1$ is achievable.
They also showed the following inherent limitations on broadcasting with throughput~$1$: no such algorithm can be fair for a system of at least two nodes against leaky-bucket adversaries, queues may need to be $\Omega(n^2+\text{burstiness})$, and broadcast algorithms need to be adaptive.

Anantharamu et al.~\cite{AnantharamuCKR-INFOCOM10} studied packet latency of broadcasting on adversarial multiple access channels by deterministic distributed algorithms when injection rates was less than~$1$.
This was continued by Anantharamu et al.~\cite{AnantharamuCKR-SIROCCO11} who considered adversarial queuing on multiple-access channels when the adversary can jam the channel.

Anantharamu and Chlebus~\cite{AnantharamuC-TCS15} considered the adversarial model in which the adversary can activate otherwise passive and anonymous nodes by injecting packets into them.
The adversary was constrained to be able to activate at most one passive node in a round; such a node remains active as long as it has packets to broadcast.
They showed that positive injection rates could be handled with bounded packet latency; the specific magnitude of rates depends on the class of algorithms, with the injection rate of $\frac{1}{2}$ being the largest one among them.
They also demonstrated that no algorithm could achieve bounded packet latency when an injection rate is greater than $\frac{3}{4}$.


\Paragraph{Related work.}

Adversarial queuing has proved to be a viable methodology to represent stability of communication algorithms  without statistical assumptions about packet injection.
This methodology was first applied to store-and-forward routing.
Borodin et al.~\cite{BorodinKRSW01} proposed this in the context of work-preserving routing algorithms when packets are routed along paths stipulated by the adversary at the time of injection, so that routing is restricted to a scheduling policy.  
The subsequent work by Andrews et al.~\cite{AndrewsAFLLK01} concentrated on the notion of universal stability, which for a scheduling policy means stability in any network, and for a network denotes stability of an arbitrary scheduling policy, both properties to hold under injection rates smaller than~$1$.
Lotker et al.~\cite{LotkerPR04} demonstrated that every work-conserving scheduling policy is stable if the injection rate is suitably small, as determined by the length of the longest path traversed by a packet.
FIFO as one of the most popular scheduling policies was studied extensively, see~\cite{BhattacharjeeGL04,CholviE07}.
Ros\' en and Tsirkin~\cite{RosenT06} considered routing against rate-$1$ adversaries; they defined reliability of an algorithm to mean that each packet is eventually delivered and  showed that reliability is achievable only in networks with no cycles of length at least~$3$. 
\'Alvarez et al.~\cite{AlvarezBDSF05} applied adversarial models to capture phenomena related to routing with varying priorities of packets and to study their impact on universal stability.
Andrews and Zhang~\cite{AndrewsZ98} gave a universal algorithm to control traffic when nodes operate as switches that need to reserve the suitable input and output ports to move a packet from the input port to the respective output one.
\'Alvarez et al.~\cite{AlvarezBS-ICPADS04} addressed the  stability of algorithms in networks with links prone to adversarial failures.
Andrews and Zhang~\cite{AndrewsZ03} proposed suitable adversarial models for networks in which nodes represent  switches connecting inputs with outputs so that routed packets encounter additional congestion constrains at nodes when they compete with other packets for input and output ports. 
Andrews and Zhang~\cite{AndrewsZ07} studied routing in wireless networks when data arrivals and transmission rates are governed by an adversary.
Blesa et al.~\cite{BlesaCFLMSST09} extended the adversarial model of wired networks to capture variability of speeds of links and sizes of packets.

Static broadcasting in multiple-access channels was considered by Greenberg and Winograd~\cite{GreenbergW-JACM85}, Koml\'os and Greenberg~\cite{KomlosG-TIT85}, and more recently by Kowalski~\cite{Kowalski05}.
Bie\'nkowski et al.~\cite{BienkowskiKKK-STACS10} and Czy\.zowicz et al.~\cite{CzyzowiczGKP11} considered algorithmic solutions on multiple-access channels for distributed-computing primitives like consensus and mutual exclusion.
Chlebus et al.~\cite{ChlebusKL-DC06} and Clementi et al.~\cite{ClementiMS02} studied the problem of performing independent idempotent tasks by distributed algorithms when processors communicate over a multiple-access channel.
The problem of waking up a multiple-access channel was first considered  by G\k asieniec et al.~\cite{GasieniecPP-JDM01}; see \cite{DeMarcoK13, DeMarcoPS07, JurdzinskiS02} for later work.

Gilbert et al.~\cite{GilbertGN09} proposed to model disruptive interference in multi-channel single-hop networks by a jamming adversary. 
This was further investigated in a number of papers including the following ones.
Dolev et al.~\cite{DolevGGN07} considered restricted gossiping in which a constant fraction of rumors needs to be disseminated when the adversary can disrupt one frequency per round.
Gilbert et al.~\cite{GilbertGKN-INFOCOM09} studied gossiping in which the adversary can disrupt up to $t$ frequencies per round and eventually all but $t$ nodes learn all but $t$ rumors.
 Dolev et al.~\cite{DolevGGKN09} considered synchronization of a multi channel under adversarial jamming.


\Paragraph{Document structure.}

This paper is organized as follows.
We summarize and review technical preliminaries in Section~\ref{sec:technical}.
Impossibility results and lower bounds are all grouped together in Section~\ref{sec:impossibilities}.
The two cases of non-adaptive algorithms for channels with collision detection and adaptive algorithms for channels without collision detection are presented in Section~\ref{sec:two-algorithms}.
A non-adaptive algorithm for channels without collision detection is given in Section~\ref{sec:non-adaptive-no-collision-detection}.
We conclude with final remarks in Section~\ref{sec:conclusion}.

\section{Technical Preliminaries}

\label{sec:technical}

Multiple access channels are specialized communication networks that support broadcast through their architecture.
In this section, we specify what properties of a communication network make it a multiple access channel.
Next, we define adversarial models, classes of broadcast algorithms and performance metrics.

We consider networks that operate in a \emph{slotted time}, which means that an execution of an algorithm is a sequence of rounds.
All nodes begin and end a round at the same time.
These stipulations determine that a network operates synchronously, as if the nodes were coordinated by a global clock.

A node may choose either to transmit or pause in a round.
Everything that a node transmits in a round is called \emph{message}.
Intuitively, messages overlapping in time collide with each other and none can be successfully received.
We assume that messages are big enough so that when two nodes transmit messages in the same round then the transmissions overlap in time.


\Paragraph{Multiple access channels.}

A node receives a feedback from the channel in each round.
A message successfully received by a node is \emph{heard} by the node.
A broadcast system is a \emph{multiple-access channel} when a message is heard by a node if and only if it is the only message transmitted in this round by any node in the network.
This implies that at most one message per round can be heard, so that the throughput rate of a multiple access channel is at most~$1$.
We consider the following three cases determined by the  multiplicity of transmissions in a round in order to introduce additional terminology.
\begin{quote}
\begin{description}
\item[\rm There are no transmissions in a round:] nodes receive  \emph{silence}  from the channel as feedback.
Such a round is called \emph{silent}.

\item[\rm There is one transmission in a round:]
 the message is heard by all the nodes in the same round.
 
\item[\rm More than one transmissions in a round:] no node can hear any message.
Multiple transmission in the same round result in a conflict for access to the channel, which we   call a \emph{collision}. 
\end{description}
\end{quote}
Transmitting and listening to the channel are considered independent activities, in that a transmitting node obtains the same feedback from the channel as a node that does not transmit in the round.
The channel is \emph{with collision detection} when the feedback from the channel allows the nodes to distinguish between silence and collision, otherwise the channel is \emph{without collision detection}.
If no collision detection mechanism is available, then nodes obtain the same feedback from the channel during a collision as during a silent round.

We define a round  \emph{void} when no packet is heard in this round.
When a round is void then either it is silent or a collision occurs in it.


\Paragraph{Adversaries.}

An adversary is defined by a set of allowable patterns of injections of packets into nodes.
An adversary generates a number of packets in each round and assigns to each packet  the node into which the packet is injected. 
We first recall the standard definition of globally constrained window-type adversaries.
Next we constrain adversaries further by  separately specifying how packets are injected into each node.

A (globally constrained) \emph{window-type adversary} is restricted by \emph{injection rate~$\rho$} and \emph{window size~$w$}. 
These two numbers~$w$ and~$\rho$ constrain the adversary's behavior as follows: for any contiguous time interval~$\tau$ of~$w$ rounds, the adversary may inject at most~$\rho w$ packets in~$\tau$.

\emph{Window adversaries with individual injection rates} are defined as follows:
Let there be given a positive integer number $w$, which is again called \emph{window size}.
Each node~$i$ is assigned a \emph{share~$s_i$}, which is a non-negative integer.
The shares satisfy the requirement $\sum_{i=1}^n s_i \le w$.
These numbers constrain the adversary's behavior as follows:  in any time interval of~$w$ contiguous rounds, the adversary may insert up to~$s_i$ packets into node~$i$.
The \emph{local injection rate} of node~$i$ is defined to be the number $\rho_i=s_i/w$.
The \emph{aggregate injection rate} is denoted by $\rho=(\sum_{i=1}^n s_i) / w$.
We refer to such a window adversary, with individual injection rates, as being of \emph{local type  $(\langle s_i\rangle_{1\le i\le n},w)$} and of \emph{aggregate  type $(\rho,w)$}.
Let a window adversary be of local type  $(\langle s_i\rangle_{1\le i\le n},w)$ with aggregate rate~$\rho$.
For this adversary, the \emph{local burstiness at node~$i$} is defined to be its share~$s_i$, and the \emph{aggregate burstiness} is defined to be the number $\delta=\sum_{i=1}^n s_i$, so that $\rho=\delta/w$.

The notion of aggregate type is similar to that of global type as considered in~\cite{ChlebusKR09,ChlebusKR-TALG}.


\Paragraph{Broadcast algorithms.}

We consider communication algorithms that are distributed, in that they are event driven.
An algorithm is determined by specifying what constitutes a possible state of a node and what are the rules governing transitions between states.

When a property of a communication environment can be used by algorithms, to be executed in this environment, then we say that this property is \emph{known}.
The adversaries we consider are not known, which means that our algorithms are not tailored to any parameters of the adversary.

The letter~$n$ denotes the number of nodes attached to a channel, which we may call the \emph{size of the network}.
Each node has a unique integer name in the range between~$1$ and~$n$.
We assume that each node knows both its name and the number~$n$.

When multiple packets are pending transmission then some of them need to be parked at nodes waiting to be transmitted.
This occurs, for example, when multiple packets are injected simultaneously in the same node.
Each node has its packets stored in a private queue.
The capacity of such a queue is assumed to be unbounded, in that a queue can store any number of packets in principle.

A message may include at most one packet and possibly some control bits.
We do not assume an upper bound on the number of control bits.
For example, an algorithm may have nodes attach an ``over'' bit to a transmitted packet to indicate that the  node will not transmit in the next round.
A message consisting entirely of control bits is legitimate.
The contents of packets do not affect an execution of a broadcasting algorithm, in that packets are treated as abstract markers.


\Paragraph{Executions of algorithms.}

All the nodes start executing a communication algorithm simultaneously.
In the course of an execution, the \emph{state of a node} is determined by the values of its private variables, which correspond to the variables used in the code of the algorithm.
A node's action in a round of an algorithm's execution consists of the following, in the order given: 
\begin{quote}
\begin{enumerate}
\item[\rm (1)]  
the node either transmits or pauses, as determined by its state,
\item[\rm (2)] 
the node receives a feedback from the channel, in the form of either hearing a message  or collision or silence,
\item[\rm (3)]  
new packets are injected into the node, if any,
\item[\rm (4)] 
a state transition occurs in the node.
\end{enumerate}
\end{quote}
An \emph{event} in a round is defined as a vector of nodes' actions in the round, indexed by the nodes.
An \emph{execution} of an algorithm is defined as a sequence of events  occurring in consecutive rounds.
All executions we consider are infinite.

A state transition of a node is determined by  the state at the end of the previous round, the packets injected in the round, and the feedback from the channel in the round.
Such a transition involves the following operations performed by a node.
Any injected packet is enqueued in the same round.
A successfully transmitted packet is discarded.
When the queue of a node is nonempty, including the packets injected in this round, then the node obtains the next packet to broadcast by dequeuing the queue.


\Paragraph{Performance of algorithms.}

The number of packets stored in a queue in a given round is called the \emph{size of the queue} in the round. 
An algorithm is \emph{stable in an execution} when there is an upper bound on the sizes of all the queues in this execution.
An algorithm is \emph{stable against an adversary} if it is stable in each execution against this adversary.
This terminology is naturally extended for classes of adversaries, to mean stability against any adversary in the class.
Classes of adversaries are usually determined by properties of their types, like the magnitude of injection rates. 

The time spent by a packet waiting in a queue at a station is called this packet's \emph{delay}.
An upper bound on packet delay is called packet \emph{latency}.
Observe that the size of a queue in a station in a round is a lower bound on some packet's delay, because at most one packet is dequeued in a round.
It follows that an algorithm of bounded latency is stable, as latency is an upper bound on queue size.

If an algorithm is stable for a multiple access channel with the first-in first-out (FIFO) queuing discipline in each station, then it is also stable with any other queuing policy.
We use the FIFO queuing in each station in each of our algorithms.
The reasons are that this discipline minimizes latency and provides fairness even when packet latency is infinite, in the sense that no packet lingers ``at the bottom of a queue'' forever.


\Paragraph{Classes of algorithms.}

Natural subclasses of deterministic distributed broadcast algorithms in multiple access channels were defined  in~\cite{ChlebusKR09,ChlebusKR-TALG}.
We use the same classification, which we give next.
\begin{quote}
We call an algorithm \emph{adaptive} when control bits can be used in its messages; otherwise an algorithm is called \emph{non-adaptive}.

An algorithm is \emph{acknowledgment based} when a node without packets to transmit stays in an initial state and resets its state to an initial state immediately after a successful transmission.

An acknowledgment-based algorithm is  \emph{oblivious} if the decision whether a  packet is to be transmitted or not depends only on the station's name and which consecutive round it is devoted to broadcasting a currently handled packet by this station, counting rounds from the first one assigned to  this packet.
\end{quote}
A related terminology is used in the literature to contrast  algorithms that sense channel continuously, called ``full sensing,'' with  ones that do not sense  channel when not having pending packets to broadcast.
All algorithms we consider are ``full sensing,'' as our definition of broadcast algorithm stipulates that stations obtain feedback from the channel in each round and make state transitions in each round.
The term ``acknowledgment-based'' is sometimes used to mean that the algorithm has the property that when a station keeps attempting to transmit a packet and eventually succeeds, then the channel ``acknowledges'' this fact by its feedback, which triggers the station to reset its state to an initial one.

When  an oblivious algorithm is executed, then a node may ignore feedback from the channel, except for detecting whether the packet transmitted was heard on the channel, which serves as an ``acknowledgment'' from the channel.
Formally, an oblivious algorithm is determined by unbounded binary sequences assigned to nodes; 
such a sequence is called a \emph{transmission sequence}.
Different nodes, executing the same algorithm concurrently, may have different transmission sequences.
These different sequences are interpreted as follows.
If the $i$th bit of the transmission sequence of a node equals~$1$, then the node transmits the currently processed packet in the $i$th round, counting rounds from the first one when the packet was started to be processed, while a~$0$ as the $i$th bit  makes the node pause in the corresponding $i$th round.
Oblivious acknowledgment-based algorithms were considered in~\cite{ChlebusKR-TALG}.
In this paper, acknowledgment-based algorithms are defined by the property that a node begins execution in an initial state and resets its state to initial after a successful transmission; a station without pending packets loops in an initial state.

An algorithm is \emph{conflict free} if in any execution at most one node transmit in any round.
An algorithm that is not conflict free is called \emph{conflict prone}.
When an algorithm is conflict free then we assume that a channel is without collision detection.


\Paragraph{Algorithm design.}

Now we give an overview of design principles of our algorithms. 

It is a natural approach to have nodes work to discover the parameters defining the adversary at hand.
The nodes could gradually improve their estimates of the shares~$s_i$ and the aggregate  maximum burstiness $\delta=\sum_{i=1}^n s_i$.
The nodes would adjust the frequency of their individual transmissions to the knowledge gained in the process of discovering the adversary.
The aggregate burstiness~$\delta$ is a lower bound on the window~$w$ of an adversary.
Observe that the aggregate burstiness in the case of injection rate~$1$ equals the window size.

For a given window adversary of local injection rates, node~$i$ is \emph{active} when its share~$s_i$ is a positive number; otherwise, when $s_i=0$, the node~$i$ is \emph{passive}.
Node~$i$ has been \emph{discovered} in the course of an execution when a packet transmitted  by~$i$ has been heard on the channel.
A discovered node is clearly active.
In the context of a specific execution, a node that has not  been  discovered by a given round is called a \emph{candidate} in this round.
A passive node is doomed to be a candidate forever.
We describe two data structures used to schedule transmissions and have nodes discover their shares.

In one approach, each node~$i$ has a private array~$C_i$ of $n$ entries.
The entry $C_i[j]$ stores an estimate of the share $s_j$ of node~$j$, for $1\le j\le n$.
Every node~$i$ will modify the entry~$C_i[k]$ in a round in the same way; therefore we may drop the indices and refer to the entries of the arrays as~$C[j]$ rather than as~$C_i[j]$.
The array~$C$ is initialized to $C[j]=0$ for every~$1\le j\le n$.
The sum $\gamma=\sum_{i=1}^n C[i]$ is an estimate of the aggregate burstiness $\delta=\sum_{i=1}^n s_i$.

The nodes running an algorithm keep adjusting the estimates stored in the array~$C$.
When node~$i$ enters a state implying $s_i>C[i]$ then $i$ considers itself \emph{underestimated}.
Detecting underestimation is implemented by having every node keep track of the transmissions in the past $\gamma$ rounds.
When a node~$i$ detects that some $k>C[i]$ packets have been injected within the respective $\gamma$ consecutive rounds, where $k$ is maximum with this property so far, then $i$ decides  in this round that it is \emph{underestimated by the amount $k-C[i]$}.

Another approach to schedule transmissions and discover shares is to have each node~$i$ use a list~$D_i$ of bits.
Such a list~$D_i$ has its terms listed as a sequence $\langle d_i(1), d_i(2),\ldots,d_i(\ell)\rangle$. 
The length~$\ell$ may be modified, but it is the same in all nodes, and additionally the inequality $\ell\le w$ holds at all times.
The lists are initialized to be empty.

The number $\ell$ will be an estimate on the burstiness of the adversary, similarly as the number $\gamma$ is for the array~$C$.
The number of occurrences of $1$ in $D_i$ has the same meaning as $C[i]$ and so is interpreted as the current estimate of the share of~$i$.
The lists are maintained so as to satisfy the following invariant: 
for each $1\le j\le \ell$, when $d_i(j)$ has been determined for all~$i$, then $d_i(j)=1$ for precisely one~$i$,  and $d_k(j)=0$ for any other value~$k\ne i$ such that $1\le k\le n$. 
When lists $D$ are only used and arrays $C$ are not explicitly maintained, then $C[i]$ will denote the number of occurrences of~$1$ in~$D_i$.
We will use $\ell$ and $\gamma$ interchangeably when lists $D_i$ are used.

Our algorithms will have the property that $C[i]$ is increased at most by the amount that node~$i$ considers to be underestimated by.
This conservative mechanism of estimating the shares is safe, in that the shares are never overestimated,  as we show next.


\begin{lemma}
\label{lem:domination-of-array-c}

If $C[i]$ is increased at most by the amount that node~$i$ considers to be underestimated by, then the inequality $C[i]\le s_i$ holds at all times, for every $1\le i\le n$.
\end{lemma}

\begin{proof} 
We show that the invariant ``the inequality $C[i]\le s_i$ holds for every $1\le i\le n$'' is preserved throughout the execution.
Initially $C[i]=0\le s_i$, for every $1\le i\le n$, so the invariant holds true in the first round.
Observe that if $C[i]\le s_i$, for every $1\le i\le n$, then also the inequalities
\[
\gamma=\sum_{i=1}^n C[i]\le \sum_{i=1}^n s_i=\delta\le w
\]
hold true.
The adversary can inject \emph{at most~$s_k$} packets into a node~$k$ during any segment of~$\gamma$ contiguous rounds, as the inequality $\gamma\le w$ means that any segment of $\gamma$ contiguous rounds could be extended to one of $w$ rounds of the execution.
A node~$k$ can detect to be underestimated by at most $s_k-C[k]$ in any round, while the invariant holds.
This makes the invariant still hold in the next round, because an update of $C[k]$ raises $C[k]$ to a value that is at most~$s_k$.
\end{proof} 

The algorithms we develop will be such that the assumptions of Lemma~\ref{lem:domination-of-array-c} hold true, and for them the inequalities $\gamma\le \delta\le w$ hold as well.
We will use arrays $C$ when each node knows the node for which an update of the current estimate of its share is to be performed.
It may happen that such knowledge is lacking: an underestimated node transmits successfully but the other nodes do not know which node transmitted. 
This is a scenario in which to use lists $D$, as a transmitting node can append~$1$ at the end of $D$ while every other node appends~$0$.

Our algorithms have nodes manipulate auxiliary lists, including the lists~$D$.
For each list, there is the \emph{main pointer} associated with this list, which points at an entry we call \emph{current} for the list, when the list is nonempty.
In a round, a main pointer either stays put or it is advanced by one position in the cyclic ordering of the entries on the list.

\section{Impossibilities and Lower Bounds}

\label{sec:impossibilities}

In this section, we present impossibility results for acknowledgment-based algorithms and lower bounds on packet latency.


\Paragraph{Impossibility results for acknowledgment-based algorithms.}

We begin by examining ac\-knowl\-edg\-ment-based algorithms. 
The action that a node performs when it begins processing a new packet is always the same, as it is determined by the initial state.
Such an action maybe either to transmit or to pause in the first round.


\begin{lemma}
\label{lem:ack-based}

Let us consider an execution of an acknowledgment-based algorithm.
If there is a node that pauses in the first round after starting to process a new packet, then, for any number $\rho >\frac{1}{2}$, the algorithm is unstable for some adversary with aggregate rate~$\rho$.
\end{lemma}

\begin{proof}
Suppose that some node $p$ pauses in the first round of processing  a new packet.
Consider an adversary who injects only into such a node $p$ as often as possible subject to an individual injection rate that is between $\frac{1}{2}$ and $\rho$.
This results in an execution in which a packet is heard not more often than every second round, while the aggregate rate is greater than $\frac{1}{2}$, so the queue at~$p$ grows unbounded.
\end{proof}


\begin{theorem}
\label{thm:ack-based-2-nodes}

Any acknowledgment-based algorithm is unstable in the multiple-access channel with collision detection that consists of just two nodes, when executed against some window adversary of burstiness $2$ and aggregate injection rate~$1$.
\end{theorem}

\begin{proof} 
Consider an ac\-knowl\-edg\-ment-based algorithm for two nodes $p$ and~$q$.
Suppose, to arrive at a contradiction, that the algorithm is stable for the aggregate injection rate~$1$.
This implies, by Lemma~\ref{lem:ack-based}, that a node transmits a new packet immediately. 

We define an execution of the algorithm with an infinite sequence of rounds  $t_1, t_2,\ldots$ determined so that there are at least $i$ void rounds by round~$t_i$.
(Recall that a round is void when no packet is heard in this round.)
The adversary will inject two packets in each odd-numbered round, a packet per node, and it will inject no packets in even-numbered rounds. 
This means that each node has its individual injection rate equal to~$\frac{1}{2}$.
The aggregate injection rate of the adversary is~$1$, and so there will be at least $i$ packets in queues in round~$t_i$.
Next we stipulate how such a scenario can occur.

Let us set $t_1$ to be the first round.
This round is silent as the nodes did not have any packets in round zero.
A collision occurs in the second round.
This is because each node got a new packet in the first round, so it transmitted it immediately. 
Let us define $t_2$ to be the second round. 
As our construction continues, each round number~$t_i$,  for $i>1$, will be even.
The adversary will not inject packets in these specific rounds $t_i$, as packets will not be injected in any even-numbered round, as a general rule.

Suppose the execution has been determined through an even-numbered round~$t_i$.
If $t_i+1$ is void then define $t_{i+1}=t_i+2$.
Otherwise some node, say $p$, transmits in round~$t_i+1$.
The node~$p$ will continue transmitting for as long as it has packets, because it transmits a new packet immediately.
If in one such a round~$t$ the node~$q$ transmits concurrently with~$p$, then this results in a collision and we define $t_{i+1}$ to be $t$ if $t$ is even or $t+1$ otherwise.
Let us suppose that this is not the case, so $p$ keeps transmitting alone. 
The node~$p$ will not have a packet to transmit in some round after $t_i+1$, since its injection rate is $\frac{1}{2}$ and so smaller than~$1$; let $t'>t_i+1$ be the first such a round.
Observe that $t'$ has to be an odd number, as otherwise a new packet would have got injected into~$p$ in the round $t'-1$ resulting in $p$ having a packet to transmit in round~$t'$.
If $q$ does not transmit in round~$t'$, then this round~$t'$ is silent and we are done; in this case we define $t_{i+1}$ to be $t'+1$, as such a round number is to be even.
Otherwise, the node~$q$ transmits in round~$t'$ successfully.
Simultaneously  both nodes obtain new packets in round~$t'$, so each has at least one available packet at the end of this round.
Each node transmits in round $t'+1$, as for each of them it is a first round of processing a new packet.
This is because $p$ did not have a packet in round~$t'$ and $q$ transmitted in round~$t'$.
Define $t_{i+1}$ to be the void round $t'+1$.
This completes the inductive construction of the sequence~$t_i$, and by this produces a contradiction with the assumption that the algorithm is stable.
\end{proof} 

Next we investigate how large could an injection rate be, as a function of the number of nodes attached to a multiple access channel, that can be handled in a stable manner by an acknowledgment-based algorithm.
It was shown by Chlebus et al.~\cite{ChlebusKR-TALG} (Theorem 5.1) that an oblivious acknowledgment-based algorithm cannot be stable when the global injection rate is at least as large as $\frac{3}{1+\lg n}$, for $n\ge 4$ nodes, where $\lg n$ denotes the binary logarithm of~$n$.
We show a related result for individual injection rates and acknowledgment-based algorithms that are not necessarily oblivious.
The following theorem was inspired by the related result in~\cite{ChlebusKR-TALG} and its proof is similar to that given in~\cite{ChlebusKR-TALG}.


\begin{theorem}
\label{thm:ack-based-n-nodes}

If an acknowledgment-based algorithm is executed by $n\ge 4$ nodes on a multiple-access channel with collision detection, then the system is unstable against an adversary for which only two nodes have positive shares, one such share is equal to~$1$ and the other to~$2$, and the window size is larger than $\lceil \lg n\rceil$. 
\end{theorem}

\begin{proof}
Let $\cA$ be a specific acknowledgment-based algorithm for the $n$ nodes. 
For each node $p$, consider an execution of algorithm $\cA$ when $p$ starts to work on a new packet.
Let us consider an execution, understood as an experiment, such that when $p$ transmits then $p$ hears a collision and when $p$ does not transmit then the round is silent.
The purpose of this experiment is to determine a sequence of bits $s(p)= (b_1, b_2,\ldots)$ such that $b_i=1$ when $p$ transmits in the $i$th round and $b_i=0$ when $p$ pauses in the $i$th round in this  execution, where we count  rounds from the first one when $p$ begins to process the packet.
Let $s(p,i)$ be this sequence truncated to the first $i$ terms, for $i\ge 1$.
There exist two nodes, $p_1$ and $p_2$, for which $s(p_1,\lceil \lg n\rceil - 1)=s(p_2,\lceil \lg n\rceil - 1)$ by the pigeonhole principle, because $\lceil \lg n\rceil-1\ge 1$ for $n\ge 4$ and  $2^{\lceil \lg n\rceil-1}=2^{\lceil \lg n\rceil}/2 < n$.

If the two sequences $s(p_1)$ and $s(p_2)$ are identical then it is sufficient for the adversary to simultaneously inject one packet into $p_1$ and another into $p_2$ and these packets will never be heard on the channel.
Then instability follows, because packets subsequently injected into $p_1$ and $p_2$ will need to be queued and none of them will ever be heard.
Otherwise, let $k\ge \lceil \lg n\rceil$ be the first position in which the sequences $s(p_1)$ and $s(p_2)$ differ.
Without loss of generality, suppose that the $k$th term of $s(p_1)$ is~$1$ and the corresponding term of~$s(p_2)$ is~$0$.
Let~$j$ be the smallest position of the sequence $s(p_2)$ such that  $j>k$ and there is~$1$ at this position in~$s(p_2)$.
We determine an adversary as follows. 
The window size $w=j$ and node $p_1$ has share~$1$ and node $p_2$ has share~$2$, which is possible for $n\ge 4$ because then $j\ge 3$.

This adversary may inject the packets as follows.
At round number~$0$, the adversary injects into each node the number of packets equal to that  node's share.
This is followed by $w$ rounds so that two packets are successfully transmitted, one in round $k$ by node $p_1$ and the other in round $w=j$ by node~$p_2$, while node~$p_2$ still has one packet.
At round $w=j$, the adversary injects  into each node the number of packets equal to this node's share.
This makes the behavior of the channel during the rounds $w+1$ through $2w$ be the same as during the rounds from~$1$ through~$w$, with the only difference that at round~$2w$ node~$p_2$ has two packets pending transmissions.
This pattern can be iterated forever, with the result that at the round number~$\ell w$ node~$p_2$ has $\ell $ pending packets.
\end{proof}

Due to the inherent limitations of acknowledgment-based algorithms, as expressed in Theorems~\ref{thm:ack-based-2-nodes} and~\ref{thm:ack-based-n-nodes}, the algorithms we develop in this paper are not acknowledgment based.


\Paragraph{Lower bounds.}

Next we present two lower bounds on packet latency.
We begin by showing that an algorithm with performance close to optimal needs to have bounds $\Omega(n+w)$ or $\Omega\bigl(w\,\text{polylog }n)$ on packet latency.


\begin{theorem}
\label{thm:general-lower-bound}

For any broadcast algorithm for a channel with $n$ nodes and for an adversary of window~$w$  and aggregate injection rate~$1$,  packet latency provided by the algorithm  in some execution is $\Omega(w \max (1,\log_w n))$, when $w\le n$, and it is $\Omega(w)$, when $w>n$.
\end{theorem}

\begin{proof} 
The adversary can inject $w$ packets in one round because the aggregate injection rate is~$1$ so the burstiness equals~$w$.

For the case $w\le n$, we consider adversaries for which each node has a share that is either~$0$ or~$1$.
Greenberg and Winograd~\cite{GreenbergW-JACM85} considered a static version of the broadcast problem, in which $k$ packets are located initially among $k$ nodes, for some $k \le n$, at most one packet per node,  and the goal is to have all of them heard on the channel.
They showed that for any algorithm it takes $\Omega\bigl(k\,\frac{\log n}{\log k}\bigr)$ time to achieve this goal in some execution of the algorithm.
Algorithms that we use handle dynamic broadcast, but they can be applied to the static version.
Namely, a translation to the static version of broadcast is as follows: let the adversary inject $w$ packets in the first round, at most one packet per node, the algorithm needs to broadcast only these packets.
This gives the bound $\Omega\bigl(w\,\frac{\log n}{\log w}\bigr)=\Omega(w\max (1,\log_w n))$.

Next consider the case $w>n$.
Let the adversary inject $w$ packets in the first round and then execute the algorithm.
It will take at least $w$  rounds to hear these packets.
\end{proof}

The next observation is that to attain latency $o(nw)$, the algorithm has to be conflict prone.


\begin{theorem}
\label{thm:collision-free-quadratic-lower-bound}

For any conflict free algorithm, a system of sufficiently many nodes~$n$, and any sufficiently large window~$w$, there is an execution in which a packet is delayed by $\Omega(nw)$ rounds against a window adversary with window $w$ and an aggregate injection rate smaller than~$1$.
\end{theorem}

\begin{proof} 
Let us consider a specific conflict-free algorithm and its execution.
The adversary is specified in two stages.
The first stage is by declaring some node to be \emph{heavy} and the remaining ones to be \emph{mavericks}.
The share of the heavy node is~$w-2$, and one of the mavericks has its share equal to~$1$, so that the aggregate injection rate is $\frac{w-1}{w}= 1-\frac{1}{w}<1$.
At this point, the heavy node may be ``known to the algorithm'' while the maverick with a positive share will be declared in the second stage, some time later in the execution, depending on the algorithm's actions.

Consider an execution $\cE_1$ in which the adversary injects only into the heavy node, with full capacity of $w-2$ packets per $w$ consecutive rounds, and does not inject into any maverick node at all.
Let us partition any execution into disjoint segments of consecutive $\frac{(n-1)w}{2}$ rounds.
The adversary injects $\frac{(n-1)(w-2)}{2}$ packets into the heavy node during each segment.
This leaves a room for $n-1$  rounds during a segment of~$\cE_1$ that are available to locate the maverick with a positive share.
We argue that an algorithm cannot locate such a maverick without incurring $\Omega(n w)$ delay.

We specify an execution~$\cE_2$ which has the same prefix as~$\cE_1$ until the first injection into a maverick.
To this end, we consider the consecutive segments of $\cE_1$ one by one, starting from the first segment; let $S$ denote a current segment of~$\cE_1$.
We proceed through a sequence of cases, depending on how many mavericks are scheduled to transmit in~$S$.

If fewer than $n-1$ mavericks are scheduled to transmit in~$S$, then some maverick is not scheduled to transmit in~$S$ at all.
We switch to $\cE_2$ by having the adversary inject a packet into a maverick that is not scheduled to transmit in~$S$, the injection occurring at the round just before segment~$S$ is to begin.
This packet waits at least the duration of a segment, which is $\Omega(n w)$.
If such a segment~$S$ exists then the argument is completed.

Let us suppose that, alternatively, every maverick is scheduled to transmit at least once during every segment of~$\cE_1$.
If some maverick is scheduled to transmit more than once in~$S$, then the number of packets in the heavy node increases during this segment~$S$, because at least one time slot contributes to delaying the heavy node in unloading its packets; in such a case we proceed to the next segment of~$\cE_1$.
If only such segments are in~$\cE_1$, then the algorithm is not stable and packet delays are arbitrarily large.
Otherwise, suppose that there is a segment~$S$ during which every maverick is scheduled to transmit exactly once.
Let us partition $S$ into first and second halves, each of  $\frac{(n-1)w}{4}$ rounds.
We consider two sub cases next.

The first sub-case occurs when the last maverick to transmit in~$S$ is scheduled to transmit in the first half of~$S$.
We switch to~$\cE_2$ by having the adversary inject into this last maverick just after its scheduled transmission in~$S$.
This packet needs to wait at least for the duration of the second half, which is~$\Omega(nw)$.

The other sub-case occurs when the last maverick to transmit in~$S$ is scheduled to transmit in the second half of~$S$.
We switch to~$\cE_2$ by having the adversary inject into this last maverick just before segment~$S$ is to begin.
This packet needs to wait at least for the duration of the first half, which is~$\Omega(nw)$.

The considered cases and sub-cases exhaust all the possibilities, and we conclude that a packet's delay $\Omega(nw)$ may occur for the considered adversary and injection rate $1-\frac{1}{w}$.
\end{proof} 

In Section~\ref{sec:non-adaptive-no-collision-detection} we give a collision-free non-adaptive  algorithm for channels without collision detection that has $\cO(nw)$ packet latency.
The question if packet latency has to be $\Omega(nw)$ when a collision-free non-adaptive  algorithm is executed on channels without collision detection is left open.

\section{Two Algorithms of Small Latency}

\label{sec:two-algorithms}

In this section we present algorithms with packet latencies that are close to optimal against an adversary of aggregate injection rate~$1$.
The algorithms are developed for the two scenarios when either (1) collision detection \emph{is} available, in which case we do not use control bits in messages, or (2) collision detection is \emph{not} available, in which case we do use control bits in messages.

\subsection{A non-adaptive algorithm with collision detection}

We develop a non-adaptive algorithm \textsc{Non-A\-dapt\-ive-Dis\-cov\-er-Shares} which uses collision detection to bound   latency.
Algorithm \textsc{Non-A\-dapt\-ive-Dis\-cov\-er-Shares} in turn uses two algorithms \textsc{Search-Col\-li\-sion-Update} and \textsc{Cycle-Col\-li\-sion-Update} executed in sequence. 
Algorithm \textsc{Search-Col\-li\-sion-Update} is executed first and next we switch to \textsc{Cycle-Col\-li\-sion-Update}  after a certain number of collisions is reached in an execution of \textsc{Search-Col\-li\-sion-Update}. 

A node $i$ uses the list of bits~$D_i$ to implement its estimate of shares.
The nodes manipulate the lists according to the principles given at the end of Section~\ref{sec:technical}, we recall them next.
Each such a list~$D$ is initially empty. 
There is a (main) pointer associated with each list, which points at the current entry.
The pointer is set to the first position on a list, when the first entry is inserted into the list, and after that it is advanced in the circular order of the list  by one position in each round.

The lists $D$ are used to schedule transmissions as follows: a node~$i$  transmits in a round when its current entry in~$D_i$  is~$1$ and otherwise it pauses.
A node $i$ that upgrades its share appends an entry with~$1$ to the end of its list $D_i$, while at the same time  all the other nodes append $0$'s to their lists. 
This provides conflict freeness, because the following invariants about lists $D_i$, for $1\le i\le n$ are maintained in each round:
\begin{enumerate}
\item[(i)]
the lists at the nodes are all of the same length,
\item[(ii)]
the main pointers of all lists are on the same position in their respective lists,
\item[(iii)]
at most one main pointer indicates a current entry with~$1$ while all the remaining pointers indicate at~$0$'s.
\end{enumerate}

Each of algorithms \textsc{Search-Col\-li\-sion-Update} and \textsc{Cycle-Col\-li\-sion-Update} has two threads, the main one and the update one, but they are implemented  and cooperate differently.


\Paragraph{Algorithm \textsc{Search-Col\-li\-sion-Update}.}

Now we explain in detail how algorithm \textsc{Search-Col\-li\-sion-Update} works.
The algorithm starts by invoking the main thread.
The main thread occasionally calls the update thread, when it is needed.
The update thread uses a binary search among the nodes to locate processors that consider themselves underestimated.
Next we describe the threads in detail.

The main thread uses the lists~$D$.
This prevents conflicts for access to the channel, as long as there are no transmissions beyond those schedule by the lists~$D$.
While the main thread is executed, some nodes may detect that there are underestimated.
Any such an underestimated node becomes \emph{persistent}, which means it will work to create a collision in order to gain an opportunity to upgrade its share. 
A persistent node transmits in each round, even when it is not scheduled to transmit by its array~$D$, as long as it has packets.
If a collision occurs, then the update thread is invoked, while the main thread pauses.
This means that the main pointers associated with the lists~$D$ are not advanced and persistent nodes do not transmit, except as participants of the binary search performed by the update thread. 
The main thread resumes after the update thread terminates.

The update thread operates as follows.
It performs search over the range of nodes by referring to intervals of names of nodes. 
In each round one interval is distinguished as \emph{current}.
In a round, each node in the current interval that wants to upgrade its share transmits a pending packet.
The search starts with the interval $[1,n]$, comprising all the names, which is set as current.
If a collision is heard then the current interval is partitioned in two subintervals in a balanced way.
These intervals will be processed recursively using a stack.
The left interval is pursued first, by becoming current, while the right interval is pushed on the stack. 
If a silence is heard for some current interval then such an interval is abandoned and the next interval is popped from the stack and made current.
If a packet is heard then the node that transmitted the packet is considered as having reserved the channel to transmits a number of packets up to its share's upgrade, which then is followed by silence.
Such a transmitting node appends an  entry with~$1$ to its list~$D_i$ for each transmission of a packet during upgrading its share, while the other nodes append $0$'s each.
The silence after a sequence of transmissions makes all stations abandon the currently processed interval  and the next interval is popped from the stack and made current.
The update thread terminates after the stack becomes empty and the main thread resumes. 

This completes the specification of algorithm \textsc{Search-Col\-li\-sion-Update}.
Let us recall the notation $\gamma=\sum_{i=1}^n C[i]$ interpreted as an estimate of the aggregate burstiness $\delta=\sum_{i=1}^n s_i$, as it was introduced in Section~\ref{sec:technical}.
In what follows, when we refer to  $\gamma$ then we mean the ultimate value it attains in an execution, like in the formulations of Lemmas~\ref{lem:search-by-collisions} and~\ref{lem:upgrade-collision}, unless indicated otherwise, like, for example, in the specification of algorithm \textsc{Non-A\-dapt\-ive-Dis\-cov\-er-Shares}.


\begin{lemma}
\label{lem:search-by-collisions}

Packet latency of \textsc{Search-Col\-li\-sion-Update} is $\cO(\gamma(1+\log n))$ when it is executed  in a system of $n$ nodes.
\end{lemma}

\begin{proof} 
We observe that once all the upgrades in an execution have been performed, packet latency is proportional to the maximum number of packets queued simultaneously plus burstiness.
The burstiness experienced in an execution is at most (the final value of)~$\gamma$.
Each upgrade of a share involves the upgrade thread, which takes at most $\lceil 1+\lg n\rceil$ rounds per upgrade.
There are at most~$\gamma$ invocations of the upgrade thread, so the total number of rounds spent on upgrades is at most $\lceil \gamma(1+\lg n)\rceil$.
The number of packets injected during the rounds spent on upgrades is at most the number of these rounds plus burstiness~$\delta$, which is $\cO(\gamma (1+\log n) )$.
It follows that the maximum number of packets queued simultaneously is $\cO(\gamma (1+\log n))$, and the packet latency is also $\cO(\gamma (1+\log n))$.
\end{proof}


\Paragraph{Algorithm \textsc{Cycle-Col\-li\-sion-Update}.}

Now we explain in detail how algorithm  \textsc{Cycle-Col\-li\-sion-Update} operates.
The two threads, main and update, start simultaneously.
They work concurrently, unless stated otherwise, for example when the main thread is paused for upgrading shares.

The main thread uses the lists~$D$, similarly as the main thread of algorithm \textsc{Search-Col\-li\-sion-Update}.
This means that a node~$i$ that has a~$1$ as a current entry in its list~$D_i$ in a round transmits if it has a packet.
The update thread uses a separate list of the names of all the nodes ordered in a circular order.
It is represented by a copy in each station, with a main pointer associated with it.
A node that is current on this list is referred to as \emph{current for update}. 
A node transmits in this round when the following is satisfied: the node is current for update, it considers itself underestimated, and it has a pending packet.

The feedback from the channel after a transmission by a node current for update can be of two kinds: either a collision or a packet heard.
We consider each of them next.

If the message is heard then the transmitting node turns \emph{persistent}, which means in will transmit a packet in every subsequent round, as long as it has pending packets.
Simultaneously, the main pointer on the circular list is advanced so that the next node becomes current for update; such an advance of this pointer occurs in each round until a collision is heard.
A node that is persistent turns back to non-persistent when either its queue becomes empty or when a collision occurs or when the node becomes current for update again.

Notice that there is only one persistent node at any time.
This is because when a node becomes persistent after a successful transmission, it continues to transmit until either it exhausts its packets or a collision occurs, which terminate the property of being persistent. 
In particular, no other node can become simultaneously persistent because this would mean two simultaneous successful transmissions: one by the first persistent node and another one by a contender node.

Now we discuss the case of a collision, which is caused by either two or three concurrent transmissions.
This bound on the number of concurrent transmissions is the case because only the following cases of transmissions are possible.
First occurs when a transmitter executes the main thread, which means the transmission is determined by the $D$ list.
Second occurs when a transmitter is ready for update.
Third occurs when a transmitter is persistent.

When a collision is heard then the main thread pauses and the update threads proceeds to upgrade shares. 
First a possible transmitter that is ready for update is given a chance to transmit.
If it wants to upgrade its share then it transmits a number of times, up to its share's  upgrade,  followed by a silent round.
Next a possible persistent transmitter is given a chance to upgrade its share.
This is performed similarly, by having it transmit a number of times, up to its share's  upgrade,  followed by a silent round.
These two upgrades are performed by appending $1$'s to the list~$D$ of the transmitting node and $0$'s to such lists of the other nodes, an entry for each transmission of a packet.
Such two silent rounds occur eventually. 
The first indicates a completion of upgrade by a transmitter that is ready for update, if there is any. The other indicating a completion of upgrade by a persistent node, if there is any.
After the second of these two silent rounds occurs, both the main thread and the update one resume regular operations, by having pointers on their respective lists advance in each round.


\begin{lemma}
\label{lem:upgrade-collision}

Packet latency of algorithm \textsc{Cycle-Col\-li\-sion-Update} is $\cO(n+q + \gamma)$ when the algorithm is executed  in a system of $n$ nodes with $q$ packets in queues at the start.
\end{lemma}

\begin{proof} 
A packet delay may be attributed to either to the number of packets inherited in the queue or to the time spent waiting for an upgrade of a node,  in which a packet resides and that considers itself underestimated, or to the time spent to upgrade shares.
There are $q$ packets inherited in queues.
The time waiting to start an upgrade is at most $n$, because the list used by the upgrade thread consists of $n$ entries.
An update costs at most three void rounds, as they comprise one collision and two silences.
These void rounds for each share's upgrade happen only once. 
As there are $\gamma$ upgrades, the total packet delay is as claimed.
\end{proof}


\Paragraph{The ultimate algorithm \textsc{Non-A\-dapt\-ive-Dis\-cov\-er-Shares}.}

Algorithm \textsc{Non-A\-dapt\-ive-Dis\-cov\-er-Shares} starts by invoking algorithm \textsc{Search-Col\-li\-sion-Update}.
The estimate $\gamma$ of the aggregate burstiness is available in each round, as explained in Section~\ref{sec:technical}.
As long as the inequality $n+\gamma \ge \gamma (1+\lg n)$ holds, where $\gamma$ is understood as  changing in time, then algorithm \textsc{Search-Col\-li\-sion-Update}  is executed.
When the inequality $n+\gamma < \gamma (1+\lg n)$ starts to hold then algorithm \textsc{Search-Col\-li\-sion-Update} stops and algorithm \textsc{Cycle-Col\-li\-sion-Update} takes over.
The algorithm is well specified because the inequality $n+\gamma \ge \gamma (1+\lg n)$ holds for small values of~$\gamma$, for example for $\gamma =0$, and the inequality $n+\gamma < \gamma (1+\lg n)$ holds for sufficiently large~$\gamma$, for example for $\gamma \ge n$.


\begin{theorem}
\label{thm:collision}

Algorithm \textsc{Non-A\-dapt\-ive-Dis\-cov\-er-Shares} provides $\cO(\min(n+w,w(1+\log n)))$  packet latency  for a channel with collision detection against the adversary of window $w$ and aggregate injection rate~$1$ in a system of $n$ nodes.
\end{theorem}

\begin{proof} 
If algorithm \textsc{Cycle-Col\-li\-sion-Update} is not invoked, then $n+\gamma \ge \gamma (1+\lg n)$ holds in all rounds.
In this case, packet latency is  $\cO(\gamma(1+\log n))$ by Lemma~\ref{lem:search-by-collisions}.

When \textsc{Cycle-Col\-li\-sion-Update} is invoked, then $n+\gamma < \gamma (1+\lg n)$ holds, when $\gamma$ is understood to denote the value of this parameter at the round of invocation of \textsc{Cycle-Col\-li\-sion-Update}.
This is the first round when the inequality holds, so then $\gamma=\Theta(n/\log n)$.
There are $\cO(n+\gamma)=\cO(n)$ packets in queues in a round of invocation of \textsc{Cycle-Col\-li\-sion-Update}, because the time spent on upgrades is $\cO(\gamma (1+\lg n))$ and the experienced burstiness is $\cO(\gamma)$.
In this case, packet latency is $\cO(n+q + \gamma)$, by Lemma~\ref{lem:upgrade-collision}, where $q=\cO(n)$ is the number of packets in queues at the start of \textsc{Cycle-Col\-li\-sion-Update}.
We obtain that $\cO(n+\gamma)$ is an upper bound on packet latency.

It follows that algorithm \textsc{Non-A\-dapt\-ive-Dis\-cov\-er-Shares} provides $\cO(\min(n+\gamma,\gamma\log n))$  packet latency.
The number $\gamma$ depends on an  execution, but the inequality $\gamma\le w$ holds, by Lemma~\ref{lem:domination-of-array-c}.
We observe that this implies $\min(n+\gamma,\gamma(1+\log n))\le \min(n+w,w(1+\log n))$.
This can be shown considering the following three cases.
If $n+w \ge w(1+\log n)$ then $n+\gamma \ge \gamma (1+\lg n)$ and $w (1+\lg n) \ge \gamma (1+\lg n)$.
If $n+w < w(1+\log n)$ and also $n+\gamma < \gamma (1+\lg n)$ then $n+\gamma > n+w$.
If $n+w < w(1+\log n)$ but $n+\gamma \ge \gamma (1+\lg n)$ then we note that $\gamma (1+\lg n) \le n+\gamma \le n+w$.
\end{proof}

\subsection{An adaptive algorithm without collision detection}

Now we give an adaptive algorithm for channels in which collision detection is not available.
We call the algorithm \textsc{Adaptive-Discover-Shares}.
This algorithm simulates the two algorithms in \textsc{Non-A\-dapt\-ive-Dis\-cov\-er-Shares}  by way of running algorithms \textsc{Search-Silence-Update} and \textsc{Cycle-Silence-Update}. 
The simulation proceeds as follows.

In both simulations of the corresponding main threads,  a node  scheduled to transmit but without pending packets in its queue transmits a control bit to indicate this, rather than pause and produce a silent round. 
It follows that a silence occurs only when a thread modified in this way creates a collision.
Such a silent round results in invoking the update thread.

The update thread in \textsc{Search-Col\-li\-sion-Update} relies on collision detection to implement a binary search.
Now we need to detect a collision among silent rounds.
Silence heard in a round~$t$ during the update thread indicates that either there was no transmissions or a collision occurred.
We resolve which of these is the case in the next $t+1$-st round  as follows.
All the nodes that transmitted in round $t$ transmit together with node~$1$.
Node~$1$ may have transmitted in round~$t$, but it transmits a message with a control bit in round $t+1$, so it does not need to have a pending packet.
There are two possible events occurring in round $t+1$: either the round is silent or a message is heard.
A silence in round~$t+1$ indicates that more than one node transmitted, as node~$1$ certainly did  transmit, so there was at least one node transmitting in the previous round~$t$.
This means that there was a collision in round~$t$, as no message was heard in that round.
If a message is heard in round $t+1$, then this only can be the message transmitted by node~$1$.
Therefore no other node transmitted in the previous round~$t$, and so round $t$ was silent.

The simulation of \textsc{Cycle-Silence-Update} proceeds as follows.
When the simulated main thread is running, then, in each round,  some node transmits a message,  as prompted by the occurrence of $1$ in its list~$D$.
Therefore a silent round indicates a collision.
When this occurs, some two nodes (a candidate for update and a persistent one) are given an opportunity to upgrade their shares.
This is performed by a sequence of transmissions in consecutive rounds, each resulting in a message heard. 
Each among these two nodes, if any, indicates a completion of this task by a silent round (when there is no corresponding node then a silent round indicates that).
Therefore two silent rounds in this situation indicate lack of transmissions in them, which triggers the two threads to resume advancing pointers on their lists.


\begin{theorem}
\label{thm:adaptive}

Algorithm \textsc{Adaptive-Discover-Shares} provides $\cO(\min(n+w,w\log n))$ packet latency for a channel without collision detection against an adversary of window~$w$  and aggregate injection rate~$1$ in a system of $n$ nodes.
\end{theorem}

\begin{proof}
The simulation we employ produces a constant overhead per each simulated round, which is verified by a direct inspection of the simulation mechanism.
Therefore packet latency of the simulating algorithms is of the same order of magnitude as that of  the simulated algorithm.
Theorem~\ref{thm:collision} gives a bound for the simulated algorithm, and the same asymptotic bound holds true for the simulating algorithm.
\end{proof}

\section{A Non-adaptive Algorithm without Collision Detection}

\label{sec:non-adaptive-no-collision-detection}

In this section we consider channels without collision detection. 
We develop a non-adaptive algorithm of bounded packet latency for aggregate injection rate~$1$.
The algorithm is called \textsc{Colorful-Nodes}.
The algorithm has each node maintain a private list of names of discovered active nodes.
The name of a newly discovered node is apparent to each node so it is immediately appended to the lists.
There is a pointer associated with the list of names of discovered nodes, which is occasionally advanced by one position in a circular manner; when this happens then all the nodes do this in unison.
It follows that these lists are identical and are manipulated in the same way by each node

An execution of  algorithm \textsc{Colorful-Nodes} begins with a stage we call preparation, which is followed by phases that are iterated in an unbounded loop.
A \emph{phase} consists of three consecutive stages.
A pure stage occurs first in a phase, it is followed by an update, and finally a makeup concludes the phase.
It may happen in some phase that the update and makeup stages are missing, to the effect that the initial pure stage of the phase comprises the whole remaining part of the execution.
Such a situation may occur only when,  starting from a certain point in time in this pure stage, a packet is heard in every round.


\Paragraph{Intuitions and motivation behind stage designs.}

The following are the intuitions that have motivated and guided the design of stages in algorithm \textsc{Colorful-Nodes}.

The purpose of the preparation stage is merely to discover at least one active node.
This stage occurs only once.

Pure stages are for transmissions by the active nodes that are already discovered.
The amount of time allotted for such transmissions is determined by the bounds on shares of the nodes as they have been estimated up to this stage.
It is during pure stages that most of the work of broadcasting is expected to be accomplished.

An update stage serves the purpose to give nodes an opportunity to announce to be underestimated.
Such announcements result in the respective entries of the array~$C$ getting incremented in order to improve the current estimates of shares.
Nodes that are not underestimated pause during an update stage, so these rounds could be considered wasted for them.
Eventually, no node is underestimated in an execution.
After this happens, update stages consist of silent rounds only, if any occur.

In any scenario, an update stage includes $n$ silent rounds.
We do not rush into an update stage when a pure stage is under way; we wait until $n$ silent rounds have been accrued during a pure stage.
Each such a silent round indicates that the corresponding node scheduled to transmit has no packets.

When a pure stage ends, then the nodes are partitioned into those that have been detected to have no packets at some point of the last pure stage and the remaining ``busy'' ones that have not been detected as such.
The silent rounds in pure and update stages create a potential for the queues to grow unbounded at nodes that consistently obtain their maximum load of packets.
It is the purpose of a makeup stage to compensate ``busy'' nodes for the rounds wasted during the forced silent rounds of the preceding pure and update stages.


\Paragraph{Details of stage implementation.}

We describe the details of implementation for each kind of stage: preparation, pure, update, and makeup.

The \emph{preparation} stage is organized such that every node has one round to transmit, assigned in a round robin manner.
A node with a packet available transmits one when the node's turn comes up, otherwise the node pauses during its time slot.
The preparation terminates after some node performs the first transmission. 
The node~$i$ whose packet has been heard during the preparation becomes discovered, which is recorded by setting $C[i]\leftarrow 1$.

During a \emph{pure} stage, the discovered nodes proceed through a sequence of transmissions, starting from the current node on the list of discovered nodes.
A node~$i$ has a segment of consecutive $C[i]>0$ rounds allotted for exclusive transmissions.
During this segment of rounds, the node~$i$ keeps transmitting as long as it has packets, otherwise the node~$i$ pauses.
The pointer is advanced, and the next node takes over, when either the current node~$i$ has used up the whole segment of~$C[i]$ assigned rounds or just after node~$i$ did not transmit while scheduled to.
After $n$ silences occur during a pure stage, then this concludes the stage and an update stage follows.

During a pure stage, a \emph{marker} is generated each time a silent round occurs.
Markers come with one of two colors.
When a node~$i$ holds a \emph{green marker}, then this color indicates that the marker was generated when $i$ was silent during a round in a segment of $C[i]$ rounds allotted for $i$ to transmit.
A \emph{red marker} held by node~$i$ indicates that it was some other node~$j$, for $i\ne j$, that was silent during a round in a segment of $C[j]$ rounds allotted to~$j$ to transmit, which generated the marker.
Every node keeps a list of markers and their assignments to nodes in its private memory.
All nodes perform operations on markers in exactly the same way in unison.
No node holds a marker in the beginning of a pure stage. 
Only nodes already discovered get markers assigned to them during a pure stage.
A discovered node may hold either no markers or a green one or a red one in a round of a pure stage.

All operations on markers are triggered by silences during update stages.
When a new marker is created and assigned to a node then some old markers may be reassigned.
The details are as follows.

Let a node~$i$ be silent in a round assigned to $i$ to transmit in an update stage: this generates a marker.
\begin{enumerate}
\item
If $i$ does not hold a marker yet, then the new marker is colored green and it is assigned to~$i$.
\item
If $i$ already holds a green marker, then the new marker is colored red and it is assigned to the first available discovered node, in the order of their names, that does not hold a marker yet.
\item
If $i$ holds a red marker, then the new marker is colored green, it is assigned to~$i$, while the original red marker held by~$i$ is reassigned to a discovered node that does not hold a marker yet.
\end{enumerate}
A pure stage terminates after every discovered node gains a marker.
A discovered node is considered \emph{colored} by the same color as the marker it holds when an update begins.
We need markers only to assign colors to discovered nodes.
After every discovered node has gained a marker, and hence a color, then we refer only to the nodes' colors.
Colors remain assigned to the discovered nodes through the end of the next makeup stage, and so to the end of the phase.
Some colors will be modified during the makeup stage.

An \emph{update} stage is to give every node one opportunity to transmit exclusively for a contiguous segment of  rounds.
This does include candidate nodes.
A node~$i$ that is underestimated by an amount~$x$ transmits $x$ times, which is followed by a silent round.
(It might happen that an underestimated node does not have sufficiently many packets ready to be used to announce by how much it is underestimated in an update stage.
This does not create a design issue, as this indicates that so far the node has had enough room to transmit its packets.)
A number~$y\le x$ of transmissions in such a situation results in an immediate increment $C[i]\leftarrow C[i]+y$ at all nodes.
In particular, when a new node~$k$ becomes discovered, then this results in setting $C[k]$ to a positive value. 
When a node simply pauses in the first round assigned to it, then the corresponding entry in the array~$C$ is not modified.
In particular, when a candidate node~$j$ does not transmit then it maintains its candidate status, which is represented by $C[j]=0$.
After each node has had a chance to perform all its transmissions in a update stage, then the stage terminates and a makeup stage follows next.

Green nodes have had their queues empty at some point in the last pure stage. 
A \emph{makeup} stage has a purpose to have red nodes empty their queues as well.
These nodes transmit in the order inherited from the list of discovered nodes, starting from the current node, if it is red, or otherwise the next red one following the current node on the list.
A red node~$i$ has a segment of consecutive $C[i]$ rounds allotted for exclusive transmissions each time it becomes current while red.
After a red node $i$ performs $C[i]$ transmissions then it maintains the red status but stops being current, unless it is the only red node in a round.
A silent round by a red node~$i$, during $C[i]$ assigned rounds, results in changing the color of the node to green immediately and advancing the pointer to the next red node, if any.
A makeup stage concludes as soon as there are no more red nodes.

This concludes the specification of all  kinds of stages, and hence of algorithm \textsc{Colorful-Nodes}.

Observe that algorithm \textsc{Colorful-Nodes} is conflict free.
This follows by the design of stages.
The relevant property is that each stage uses a list of nodes and a pointer on this list indicates which node is to transmit in a given round.


\Paragraph{The performance of algorithm \textsc{Colorful-Nodes}.}

The algorithm uses four kinds of stages.
Assuming that some packets are injected, three stages are of bounded duration while one is not necessarily so.
A preparation stage ends as soon as a packet is heard, which is a certain event assuming that some packets are injected.
An update stage takes up to $n+w$ rounds.
A makeup stage is also of bounded duration, as we will show in Lemma~\ref{lem:time-of-makeup}. 
Pure stages are different in that it is possible that some pure stage does not end.
If this is the case, then there are finitely many phases, with the last one ending in a pure stage.
In such a scenario, starting from some round, a packet is heard in every round.


\begin{lemma}
\label{lem:two-phases}

A packet stays in a queue during at most two consecutive phases.
\end{lemma}

\begin{proof}
Let us consider a packet injected into some node~$v$ in a phase~$J$.
If either phase~$J$ or phase~$J+1$ are the last phases then the packet cannot stay beyond them.

In what follows, we assume that both phases $J$ and $J+1$ end, which means that each stage in them ends.
This implies that the queue of each node becomes empty at some point during each phase.
This property is provided by the design of the algorithm and by the meaning of colors of nodes.
Green nodes get their queues empty early in a phase, during the respective pure stage, which comes to an end by the assumption in this case.
Red nodes accomplish getting their queues empty during the next makeup stage, which also comes to an end by the same assumption.
A packet is heard by the first round following its injection in which its queue becomes empty.

If the considered packet gets injected before the queue at~$v$ becomes empty during phase~$J$, then all the rounds when the packet stays in its queue are included in phase~$J$.
When the queue of $v$ gets empty in the phase~$J$ before the packet is injected in this very phase~$J$, then the packet will be heard by the end of the next phase~$J+1$ at the latest, by the round in which the $v$'s queue gets empty in phase $J+1$.
In any case, a packet injected in phase~$J$ gets heard by the end of phase~$J+1$.
\end{proof}

At most $n$ silent rounds occur during the preparation stage, starting from the first packet injection.
Similarly, there are at most $n$ silent rounds during every stage that follows.
For accounting purposes, we partition silent rounds into \emph{blocks} defined as follows.
The first block comprises the silent rounds that occur during the last $n$ rounds of the preparation stage (or all the silent rounds during the preparation stage, in case there are less than $n$ rounds in the preparation stage), and those silent rounds that occur during the first pure stage and, assuming that this pure stage ends, the silent rounds of the first update stage.
The next block, if it exists, consists of at most $3n$ silences incurred during the immediately following makeup, pure and update stages, assuming that all these stages exist.
This continues throughout an execution, a block comprising silences in consecutive makeup, pure and update stages.
When some stage does not end, then this results in some block occurring last and having fewer silent rounds than it would have otherwise.


\begin{lemma}
\label{lem:queues-in-CN}

There  are $\cO(n+w)$ packets in queues in any round of an execution of algorithm \textsc{Colorful-Nodes} against adversaries of window~$w$  and aggregate injection rate~$1$ in a channel with $n$ nodes.
\end{lemma}

\begin{proof}
First let us consider the case of execution in which each stage ends.
Such executions have infinitely many phases and infinitely many blocks.

The total net increase of the number of packets in queues during up to $3n$ silent rounds of a block is at most $3n+w$, where the number $3n$ is due to the injection rate and $w$ to the burstiness~$\delta\le w$.
We argue that the number $6n+w$ of packets is an upper bound on the number of packets queued at all times.
The argument is by induction on the phase number, where we show the invariant that $6n+w$ is an  upper bound on the number of packets in all the queues.

The silent rounds occurring in a makeup stage are accounted for in the next makeup stage, so we may consider only rounds with successful transmissions in a makeup stage for the sake of accounting purposes.
These rounds result in suppressing the number of packets in the red nodes.
While the number of packets of the red nodes is shrinking over a makeup stage, with possible variations due to burstiness, the number of packets accumulated in the green nodes could be  increasing.
This has the effect of trading (a decrease of the number of) packets in red nodes for some (increase of the corresponding number of) packets in green nodes.
This means that the increase of the number of packets in queues, due to silent rounds, by $3n+w$ packets, may affect both red and green nodes.

For a specific block of up to $3n$ silent rounds, this increase of the number of packets is neutralized in a node when its queue becomes empty.
As the neutralization process is spread over at most two consecutive phases, by Lemma~\ref{lem:two-phases}, the increases contributed by two consecutive blocks may exist simultaneously, which may result in the effect of up to doubling the increase of $3n$ packets, contributed by one block, due to injection rate, while the contribution of $w$ to the bound, due to the burstiness, is a global one-time effect.

Next consider the case of executions in which some stage does not end.
It is not the first preparation stage, unless there are no packets to transmit.
This stage that does not end would end if only sufficiently many silent rounds occurred.
It follows that, in this case, a packet is heard in every round starting from some round~$t$.
Any increase of the number of packets starting from round~$t$ can be due only to burstiness.
Clearly, there exists another execution, that proceeds in exactly the same manner as the considered one up to round~$t$, and which afterwards has the property that each stage ends.
The number of queued packets up to round $t$ is the same in both executions.
This number of packets is $\cO(n+w)$, because one of these executions has infinitely many phases and the bound $\cO(n+w)$ on the number of packets in queues applies to any round in it.
\end{proof}

A block of silent rounds may result in an increase of the number of packets in queues of red nodes, as compared to the beginning of the immediately preceding pure stage.
A makeup stage is there to alleviate this effect.


\begin{lemma}
\label{lem:time-of-makeup}

A makeup stage of algorithm \textsc{Colorful-Nodes}  takes $\cO(nw)$ rounds against adversaries of window~$w$  and aggregate injection rate~$1$ in a channel with $n$ nodes.
\end{lemma}

\begin{proof} 
We denote by $G$ and~$R$ the sets of  nodes that begin a makeup stage as green and red, respectively.
Let $g$ equal the sum $\sum_{i\in G} C[i]=g$ of the entries of the array~$C$ over the green nodes, and $r$ be the similar sum $\sum_{i\in R} C[i]=r$ of the entries of the array~$C$ over the red nodes, as at the start of the makeup stage.
We have $g+r = \gamma$, because $G$ and~$R$ make a partition of all the nodes.
By  Lemma~\ref{lem:domination-of-array-c},  the sets $G$ of green nodes and $R$ of red nodes have had packets injected into them with the cumulative rates of at most~$g/\gamma$ and~$r/\gamma$, respectively, during the previous phase.

The rounds of the makeup stage, when only the red nodes transmit, can be conceptually partitioned into disjoint segments of~$\gamma$~rounds each.
We may employ the following accounting to assign roles to rounds in each such a segment.
The first~$r$ rounds could be considered as devoted to unloading new packets injected into red nodes during this segment.
The remaining~$g$ rounds could be treated as accounting for unloading packets injected into red nodes during the preceding block of silent rounds.

There are some $p$ packets in red nodes in the beginning of this makeup stage.
We need to account for the packets injected into these nodes during at most two phases, as each node gets its queue empty at least once during a phase.
A phase contributes at most $3n$ silent rounds, which translates into at most $6n$ packets, by Lemma~\ref{lem:two-phases}.
There could be a surge of up to $w$ packets injected due to burstiness, but this fluctuation needs to be compensated by the corresponding number of rounds that follow within the window of $w$ rounds with no packets injected in them, due to the definition of a window-type adversary.
This means that such surges are compensated by the adversary's behavior and may contribute up to $w$ to packet latency.
We conclude that we can use $p=6n$ in our estimates of the duration of a makeup stage, as we account for rounds spent on transmissions compensating rounds when the adversary could inject.

It takes $\frac{p}{g}$ segments of $\gamma$ rounds each to dispose of all the packets in the red nodes, possibly increased by at most $w$ rounds due to burstiness.
This makes $\frac{p\gamma}{g}+w$ rounds of the makeup stage when packets are heard.
There are also $r$ silent rounds in this stage.
The total number of rounds is therefore $\cO(nw)$, because $p=6n$ and the inequalities $\gamma\le w$, $r\le n$, and $g\ge 1$ hold.
\end{proof}

The following theorem summarizes the performance of algorithm \textsc{Colorful-Nodes}.


\begin{theorem}
\label{thm:colorful-nodes}
 
When algorithm \textsc{Colorful-Nodes} is executed against adversaries of window~$w$  and aggregate injection rate~$1$ in a channel with $n$ nodes, then packet latency is $\cO(nw)$, while the number of queued packets is $\cO(n+w)$ in any round.
\end{theorem}

\begin{proof} 
The estimate of a bound on the number of queued packets in any round is given by Lemma~\ref{lem:queues-in-CN}.
To consider packet latency, we proceed through two cases determined by whether all stages eventually end or rather some stage lasts forever.

Suppose first that each stage eventually ends. 
We examine how stages of all kinds contribute to the length of phases and through that to packet latency.

Each of the preparation, update and makeup stages contributes at most their own duration to the length of a phase.
In the case of the preparation stage, we mean that is over in at most $n$ rounds since the first packet is injected.
An update stage takes $\cO(n+w)$ rounds, by its design.
A  makeup stage takes $\cO(nw)$ rounds, by Lemma~\ref{lem:time-of-makeup}.

Next we consider pure stages.
There is no general upper bound on the duration of any such a stage.
A pure stage ends after $n$ silent rounds.
In the course of a pure stage, the rate of transmitting packets is equal to the rate with which they were injected in the preceding phase, by how the array~$C$ is used in scheduling transmissions.
This means that a packet gets delayed by at most $n$ plus the maximum number of packets queued while this packet is handled, which is $\cO(n+w)$ by Lemma~\ref{lem:queues-in-CN}.

Next we consider the case when some stage never ends.
It needs to be a pure stage, as makeup stages are excluded by Lemma~\ref{lem:time-of-makeup}.
Up to this pure stage, a bound on packet latency derived with the assumption that each stage ends applies.
The rate of transmitting packets during the last (unbounded) pure stage is equal to the rate with which packets were injected in the preceding phase, because a packet is heard in every round, starting from some round in the pure stage.
It follows that packet latency during such an unbounded pure stage is of the order of magnitude of the number of queued packets, which is $\cO(n+w)$ by Lemma~\ref{lem:queues-in-CN}.

There are finitely many cases, as considered above, each producing a partial bound on packet delay.
Each is  either of the form $\cO(n+w)$ or $\cO(nw)$, which gives~$\cO(nw)$ as the overall bound.
\end{proof}

Observe that the bound on packet latency in Theorem~\ref{thm:colorful-nodes} is tight.
This is because algorithm \textsc{Colorful-Nodes} is conflict free, so packet latency can be $\Omega(nw)$ in some executions when the algorithm is executed against  adversaries of sufficiently large windows~$w$ in systems of sufficiently large sizes~$n$, by Theorem~\ref{thm:collision-free-quadratic-lower-bound}.

\section{Conclusion}

\label{sec:conclusion}

We introduced a model of adversarial queuing on multiple access channels in which individual injection rates are associated with nodes.
The partitioning of the aggregate rate among the nodes constrains the adversary but is unknown to the nodes.
We developed a number of algorithms for the aggregate rate~$1$ that transmit packets with bounded latency.
The bounds on queue size and packet latency of our algorithms are not expressed in terms of the distribution of the aggregate injection rate among the nodes as their individual rates, but are given only in terms of the number of nodes~$n$ and the burstiness, which equals the window size~$w$ for the aggregate rate~$1$.

The purpose of this work was to compare the communication environments determined by multiple access channels in which adversaries are determined by individual injection rates with channels in which adversaries are constrained only by global injection rates, as studied in~\cite{ChlebusKR09}.
In both environments, no property of adversaries is known to algorithms.
A comparison of these adversarial models was to be in terms of the attainable quality of broadcasting, for the maximum throughput of~$1$.
The main discovered difference between the globally-restrained and individually-restrained adversaries is that bounded packet latency  by non-adaptive  algorithms is achievable in an adversarial model in which individual injection rates are associated with nodes, which is impossible for adversaries that are globally-restrained only.

We developed algorithms  for a window-type adversary with packet latency that is close to asymptotically optimal in the following two cases.
One is when the algorithms are adaptive and channels are without collision detection.
Another is when algorithms are non-adaptive but channels are with collision detection.
Packet latency of non-adaptive algorithms for channels without collision detection turned out to be more challenging to restrict.
The algorithm we developed achieves $\cO(nw)$ bound on packet latency.
This algorithm avoids conflicts for access to the channel.
As we showed,  packet latency has to be $\Omega(nw)$ for such conflict-avoiding algorithms.
This means that the developed algorithm is best possible in this class in terms of asymptotic packet latency.

The question if a non-adaptive algorithm can achieve packet latency that is asymptotically smaller than~$nw$, for channels without collision detection and for window-type adversaries of individual injection rates, remains open. 

The adversarial model considered in this work is of the window type.
It is a natural question to ask how to extend the adversarial model of individual injection rates to the general leaky-bucket case, and how would such an adversarial model affect algorithms' performance.
In the case of globally-constrained adversaries with injection rate~$1$, it was shown in~\cite{ChlebusKR09} that the two models of window and leaky-bucket adversaries make a difference even for small size systems.


\bibliographystyle{abbrv}

\bibliography{mac-ind-inj-rates}

\end{document}